\documentclass[a4paper, 10pt]{article}
\usepackage{amsmath,amsthm,amsfonts}
\usepackage[margin=30mm]{geometry}
\usepackage{natbib}
\usepackage{hyperref}
\newtheorem{thm}{Theorem}

\newtheorem{lemma}{Lemma}
\usepackage{tikz}
\usetikzlibrary{calc}
\usepackage{subcaption}
\usepackage[ruled,vlined]{algorithm2e}
\tikzset{
l/.style={gray,dashed},
ll/.style={red,dashed},
grn/.style=green!70!black,
b/.style=blue,
r/.style=red,
p/.style=purple}
\usepackage{color}
\newcommand{\Z}{\mathbb{Z}}
\renewcommand{\P}{\mathbb{P}}
\newcommand{\X}{\mathcal{X}}
\newcommand{\R}{\mathbb{R}}
\newcommand{\D}{\mathcal{D}}

\newcommand{\red}{\textcolor{red}}
\title{Hypothesis Testing for Topological Data Analysis}

\author{Andrew Robinson \& Katharine Turner}

\begin{document}
\maketitle

\section*{Abstract}

Persistent homology is a vital tool for topological data analysis. Previous work has developed some statistical estimators for characteristics of collections of persistence diagrams.  However, tools that provide statistical inference for  observations that are persistence diagrams are limited.  Specifically, there is a need for tests that can assess the strength of evidence against a claim that two samples arise from the same population or process.  We propose the use of randomization-style null hypothesis significance tests (NHST) for these situations.  The test is based on a loss function that comprises pairwise distances between the elements of each sample and all the elements in the other sample.  We use this method to analyze a range of simulated and experimental data. Through these examples we experimentally explore the power of the $p$-values.  Our results show that the randomization-style NHST based on pairwise distances can distinguish between samples from different processes, which suggests that its use for hypothesis tests upon persistence diagrams is reasonable.  We demonstrate its application on a real dataset of fMRI data of patients with ADHD.

\section{Introduction}

Topological Data Analysis (TDA) focuses upon considering shape within data. For example, samples may lie on a submanifold and we may want to learn about this manifold, or we may want to understand the higher dimensional correlations between different variables. The main tool used in TDA is persistence homology, which summarizes how the topology changes through a filtration of a space. An important use of TDA is as a preprocessing tool; getting a topological summary of each object that may be more tractable than the raw information. It may highlight geometric and topological features that are of particular interest. Examples of applications include analysis of the shape of human jaws \citep{gamble2010exploring}, plant root systems \citep{bendich2010computing}, shapes of calcanei bones of various primates \citep{turner2014persistent} and retrieval of trademark symbols \citep{cerri2006retrieval}.



The need to introduce statistical techniques to topological data analysis has become increasingly apparent. The observations used as input to most analysis techniques are generated stochastically, often using a sampling mechanism from a process or population. Therefore, it is useful to talk not only about a single persistence diagram, but about an entire collection of them, i.e.\ a sample, drawn from some distribution or process. Some progress has been made in, for example, calculating means and variances, and applying statistical inference techniques \citep{mileyko2011probability, turner2014frechet, balakrishnan2013statistical, turner2013medians, chazal2013optimal, bubenik2007statistical}. Alternative approaches involve reinterpreting the persistence diagrams as some functional summary lying in a larger Hilbert space. In particular, there has been work on performing statistics with persistence landscapes, including randomization tests \citep{bubenik2015statistical}. Another functional summary used is the persistent homology rank function \citep{Robins2015}. Here, we focus on the situation in which two sets of samples of persistence diagrams have each been derived from a process, and our goal is to assess the strength of evidence against the assertion that the processes are the same.

Null hypothesis significance testing (NHST) is a commonly used and important statistical tool that provides a measure of the strength of evidence against a hypothesis. In the current setting, NHST will quantify the the differences between two different types of underlying objects or processes, using persistence diagrams as observations. For example, NHST  can provide a necessary condition as to whether particular persistence diagrams could be used for classification. 

Unfortunately, the space of persistence diagrams is geometrically very complicated. It is infinite in dimension and arbitrarily curved \citep{turner2014frechet, turner2013medians}. As a result it is not plausible to use any parametric models for distributions, so we cannot do NHST using a method that requires an assumption of an underlying parametric model. Our approach is to instead find a relevant joint loss function and then use a randomization test (also known as a permutation test). This method of NHST is standard in statistical theory and is theoretically rigorous \citep[see, e.g.,][]{casella+berger-1990, welsh-1996}. 

The theory behind the randomization test ensures that when the two sets of diagrams are drawn from the same single distribution of diagrams, then the p value obtained is a random variable with a uniform distribution over an evenly spaced subset of $[0,1]$. However we do not know of any theory that the p value will necessarily be low if the distributions are different.  Furthermore, since persistence diagrams are summary statistics, it is possible that the distributions of the underlying objects under analysis are different but the corresponding distributions of persistence diagrams are similar. We will therefore show by example that there do exist situations where the null hypothesis might be correctly rejected by our method.

In the following section we provide a brief overview of the background theory for TDA and the rationale behind NHST. We then develop a test procedure in Section \ref{sec:procedure} and also outline a Monte Carlo simulation to estimate the corresponding $p$-values in Section \ref{sec:Monte}.   
In Section \ref{sec:examples} we apply the resulting algorithm to a range of data including point clouds of shapes and the persistent homology transform of silhouette data, and the concurrence filtration for fMRI data, respectively. 

%
%
%

  
\section{Preliminaries}

\subsection{TDA Background Theory}

Persistence diagrams are summaries of how the homology groups evolve over filtered spaces. Homology can be computed over any ring but persistent homology requires a field. For computational purposes this field is usually $\Z_2$. In this paper all homology will be computed over $\Z_2$.

A \emph{$k$-simplex} is the convex hull of $k+1$ affinely independent points $v_0,v_1, \ldots v_k$ and is denoted $[v_0,v_1,\ldots,v_k]$. For example, the $0$-simplex $[v_0]$ is the vertex $v_0$, the $1$-simplex $[v_0,v_1]$ is the edge between the vertices $v_0$ and $v_1$ and the $2$ simplex $[v_0, v_1, v_2]$ is the triangle bordered by the edges $[v_0,v_1]$, $[v_1, v_2]$ and $[v_0, v_2]$. Technically, there is an orientation on simplices. If $\tau$ is a permutation then $[v_0,v_1,\ldots,v_k] = (-1)^{\operatorname{sgn}(\tau)}[v_{\tau(0)}, v_{\tau(1)}, \ldots , v_{\tau(k)}]$. However, if we are considering homology over $\Z_2$ then $1=-1$ and we can ignore orientation.

We call $[u_0, u_1, \ldots u_j]$ a \emph{face} of $[v_0,v_1, \ldots v_k]$ if $\{u_0, u_1, \ldots u_j\}\subset \{v_0,v_1, \ldots v_k\}$. A \emph{simplicial complex} $K$ is a countable set of simplices such that
\begin{itemize}
\item Every face of a simplex in $K$ is also in $K$.
\item If two simplices $\sigma_1,\sigma_2$ are in $K$ then their intersection is either empty or a face of both $\sigma_1$ and $\sigma_2$.
\end{itemize}

Given a finite simplicial complex $K$, a \emph{simplicial $k$-chain} is a formal linear combination (with coefficients in our field of choice) of $k$-simplices in $K$. The set of $k$-chains forms a vector space $C_k(K)$. We define the boundary map $\partial_k:C_k(K) \to C_{k-1}(K)$ by setting
\begin{displaymath}
\partial_k([v_0, v_1, \ldots v_k]) = \sum_{j=0}^k (-1)^j[v_0,\ldots \hat{v_j}, \ldots v_k]= \sum_{j=0}^k [v_0,\ldots \hat{v_j}, \ldots v_k]
\end{displaymath}
for each $k$-simplex and extending to $k$-chains linearly. The second equality follows because our coefficient group is $\Z_2$ where $-1=1$.

Elements of $B_k(K) = \operatorname{im} \partial_{k+1}$ are called boundaries and elements of $Z_k(K)=\operatorname{ker} \partial_k$ are called cycles. Direct computation shows $\partial_{k+1}\circ\partial_k=0$ and hence $B_k(K) \subseteq Z_k(K)$. This means we can define the $k^{th}$ \emph{homology group} of $K$ to be
\begin{displaymath}
H_k(K):=Z_k(K)/B_k(K).
\end{displaymath}
A comprehensive introduction to homology can be found in \cite{hatcher2002algebraic}.

 A filter simplicial complex $K = \{K_r | r\in \R\}$ is a family of countable simplicial complexes indexed over the real numbers such that each $K_a$ is a simplicial complex and $K_a \subseteq K_b$ for $a \leq b$. 
We wish to describe how the topology of the filtration changes as the parameter increases. 
For $a\leq b$ we have an inclusion map of simplicial complexes $\iota: K_a \to K_b$ that induces inclusion maps 
\begin{displaymath}
	\iota:B_k(K_a) \to B_k(K_b) \quad \text{and}\quad \iota: Z_k(K_a) \to Z_k(K_b).
\end{displaymath}	
These inclusions induce homomorphisms (which are generally not inclusions) on the homology groups: 
\begin{displaymath}
	\iota_k^{a\to b}:H_k(K_a) \to H_k(K_b).
\end{displaymath} 
The image of $\iota_k^{a\to b}$ consists of equivalence classes of cycles that were present in $K_a$, where the homological equivalence is measured with respect to boundaries in $K_b$.  
%
%
%
We then define the $k$th dimensional persistence diagram as a multiset in $\{(x,y)\in [-\infty, \infty]^2: x<y\}$ such that the number of points (counting multiplicity) in $[-\infty, a]\times [b, \infty]$ is the rank of $\iota_k^{a\to b}$ for all $a<b$. We also include countably infinitely many copies of the diagonal, these represent homology classes that do not persist for any positive amount of time.  

Although we are considering how the topology is changing we actually learn a lot about geometrical features because the filtration has a quantifying effect. Figure~\ref{fig:PH} demonstrates how persistence diagrams may distinguish loops, even when they are topologically equivalent.

\begin{figure}
\begin{center}
\includegraphics[height=1.2 in]{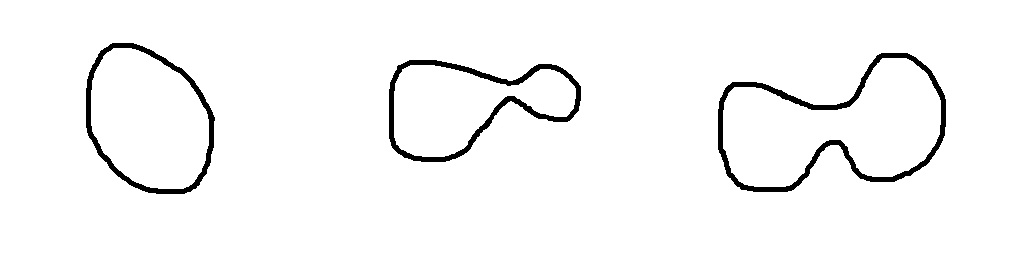}
\begin{minipage}{0.3\linewidth}\centering
\begin{tikzpicture}
[scale=.4]
\draw[->] (0,-1)--(0,7);
\draw[->](-1,0)--(7,0);
\draw[](-1,-1)--(7,7);

\fill (0,5) circle (6pt);

\end{tikzpicture}
\end{minipage}
\begin{minipage}{0.3\linewidth}\centering
\begin{tikzpicture}
[scale=.4]
\draw[->] (0,-1)--(0,7);
\draw[->](-1,0)--(7,0);
\draw[](-1,-1)--(7,7);
	
\fill (0,4) circle (6pt);
\fill (1,2) circle (6pt);
	
\end{tikzpicture}
\end{minipage}
\begin{minipage}{0.3\linewidth}\centering
\begin{tikzpicture}[scale=.4]
\draw[->] (0,-1)--(0,7);
\draw[->](-1,0)--(7,0);
\draw[](-1,-1)--(7,7);

\fill (0,4.5) circle (6pt);
\fill (2,4.5) circle (6pt);

\end{tikzpicture}
\end{minipage}
\end{center}
\caption{The dimensional one persistence diagrams for the filtrations by the distance function for three different loops.}\label{fig:PH}
\end{figure}
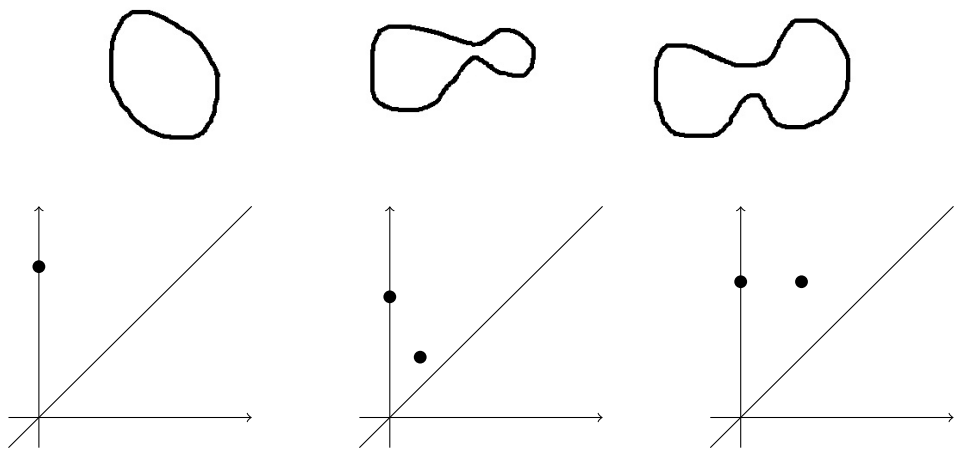

\subsection{Distance functions on the space of persistence diagrams}\label{subsec:DistPers}

Let $\D$ denote the space of persistence diagrams. There are many choices of metrics in $\D$, analogous to the variety of metrics on spaces of functions. We will be considering the distance metric that is analogous to the $L_2$ distance in the space of functions on a discrete space and the $2$-Wasserstein distance between probability distributions. A natural family of metrics is discussed in \cite{turner2013medians}.

Let $X$ and $Y$ be diagrams. We can consider bijections $\phi$ between the points in $X$ and the points in $Y$. These are the transport plans that we consider. Bijections always exist because there are countably many points at every location on the diagonal. We only need to consider bijections where off-diagonal points are either paired with off-diagonal points or with the point on the diagonal that is closest to it.

Define 
\begin{align}\label{eq:distance}
d_p(X,Y) = \left( \inf_{\phi:X \to Y} \sum_{x\in X} \|x-\phi(x)\|_p^p\right)^{1/p}
\end{align} 
where $\|x-\phi(x)\|_p^p$ is the distance from $x$ (respectively $\phi(x)$) to closest point on the diagonal whenever $\phi(x)$ (respectively $x$) is a copy of the diagonal.

We will call a bijection between points \emph{optimal} if it achieves the infimum. We can find an optimal bijection, given two diagrams $X$ and $Y$ with only finitely many off diagonal points, using the Hungarian algorithm (also known as  Munkres assignment algorithm). Suppose $X$ has $n$ off-diagonal points, labelled $x_1, x_2, \ldots x_n$, and $Y$ has $m$ off-diagonal points, labelled $y_1, y_2, \ldots y_m$. Let $x_{n+1}, x_{n+2}, \ldots x_{n+m}$ and $y_{m+1}, y_{m+2}, \ldots y_{n+m}$ be copies of the diagonal. We construct a cost matrix with $n+m$ column and rows where the $(i,j)$ entry is $\|x_i -y_j\|_2^2$. When either $x_i$ or $y_j$ is a copy of a diagonal then this is the perpendicular distance. Each transportation plan corresponds to an assignment of rows to columns --- a bijection between the points in $X$ and those in $Y$. 

Given two sets $X$ and $Y$, and pairwise costs associated to assigning to $x\in X$ the object $y\in Y$, the Hungarian algorithm finds the least-cost bijective assignment. Suppose we have two diagrams $X$ and $Y$ each with only finitely many off-diagonal points. Consider as many copies of the diagonal in $X$ and $Y$ to allow the option of matching every off-diagonal point with the diagonal. The cost of $x\in X$ doing task $y\in Y$ is $\|x-y\|_2^2$. The total cost of an assignment (or in other words bijection) $\phi$ is $\sum_{x\in X} \|x-\phi(x)\|_2^2$. The  Hungarian algorithm gives us a bijection $\phi$ that minimizes this cost. This means it gives an optimal bijection between $X$ and $Y$.

By taking the limit as $p$ goes to infinity we get the bottleneck distance between two persistence diagrams $X$ and $Y$; 
$$d_\infty(X,Y) = \inf_{\phi:X \to Y} \sup_{x\in X}\|x-\phi(x)\|_\infty .$$

There are many ways to create filtrations of interest. One of the most common ways is the forming of Rips complexes from point cloud data. We will use this method for our simulated examples later. Given a point cloud $\{x_1,x_2, \ldots x_N\}$ of points in Euclidean space $\R^n$ we define the \emph{Rips complex} with parameter $\epsilon$ (denoted $\mathcal{R}(\epsilon)$) to be the flag complex on the graph whose vertices are $\{x_1,x_2, \ldots x_N\}$ and contains the edge $(x_i, x_j)$ when $\|x_i-x_j\|\leq \epsilon$. We then build a filtration by considering the Rips complexes under an increasing parameter.

\subsection{Persistence Diagrams as random elements}

If our method of constructing a filtration is in some way random, including, for example, by the random selection of units from a population or process, then this randomness can also be seen in the corresponding collection of diagrams. A distribution of filtrations determines a distribution of diagrams. As a result we have a persistence diagram valued random element. This process works for any method of creating a filtration whether it is sublevel sets of a function or the \v Cech or Rips complexes from a point cloud.

For example, suppose we are sampling points $m$ from a subset $K$ of $\R^d$ with some noise. We are stochastically generating a point cloud that will approximate $K$. This sample generates a distribution $\rho_{\text{point clouds}}$ of sets of $m$ points in $\R^d$. Each point cloud determines a filtration of simplicial complexes and hence a distribution $\rho_{\text{filtrations}}$ of filtrations of simplicial complexes. In turn each of the filtrations determines a distribution $\rho_{K}$ of persistence diagrams. Every time we draw $m$ sample points to create a point cloud we are effectively drawing a sample from the distribution $\rho_{\text{point clouds}}$ and hence also drawing a sample persistence diagram from $\rho_{K}$. Under certain conditions, for example that the sample is random, we can learn something about $K$ by analyzing $\rho_{K}$.  Suppose we have another subset $L$ of $\R^d$ which we can similarly sample to form point clouds. We may wish to know if $K$ and $L$ are different. Our null hypothesis would be that they are the same subset. A necessary, but not sufficient, criterion for $K$ to be $L$ is that $\rho_K=\rho_L$. This implies that our null hypothesis for studying persistence diagrams is that the underlying distributions from which $\rho_K$ and $\rho_L$ are drawn are the same. We later consider a simulated examples of this form.

\subsection{Null Hypothesis Significance Testing}

We now review the algorithm and rationale behind null hypothesis significance tests.  Further reading can be found in many introductory and medium-level statistical texts; we mention for example, \citet{casella+berger-1990}, \citet{welsh-1996} and \citet{pawitan-2001}.  The steps for the test are as follows.  First, choose a parameter that represents the data in some way, and about which a pertinent hypothesis can be formed.  Popular examples of parameters include the sample mean and median for tests of location, and the variance for tests of spread.  Second, choose a statistic to use to estimate the parameter.  Third, predict the statistical behavior of the statistic under the null hypothesis, trying to capture the full range of variability that is implied by the model.  Commonly, statistical theory is used to nominate a distribution for the test statistic assuming that the null hypothesis is true. For example, the sample mean might be assumed to have a Gaussian distribution, based on the Central Limit Theorem or the assumed distribution of the data. Fourth, compare the observed value of the test statistic with the expected behavior under the null hypothesis.  If the test statistic is anomalous compared with the expected behavior under the null hypothesis, then it is considered to be evidence against the null hypothesis.

Formal approaches to testing diverge at the point of comparison, and we discuss two of them here.  One approach requires nominating a cutoff, called the \emph{size} of the test, which is by definition the probability of rejecting the null hypothesis.  That is, the cutoff is set to be the probability of rejecting the null hypothesis when it is true.  Then, the probability of observing a result as or more extreme than the observed test statistic is computed based on repeated experiments, assuming that the null hypothesis is true.  That is, we imagine a set of identical experiments to be carried out, for which the null hypothesis is true, and ask what is the proportion of that set for which the computed test statistic is more extreme than the observed value, relative to the null hypothesis.  Then we report the outcome of the comparison of the observed probability against the size of the test.  This is a Neyman--Pearson approach to testing, and in the classical case where the variance is unknown and the hypothesis concerns the mean, the reference distribution is Student's $t$ distribution, with degrees of freedom equal to the sample size $n$ minus one.  Another approach, following Fisher, simply reports the estimated probability computed above, called the p value. The reader is free to place their own interpretation on the p value.  We will follow the latter approach.

In any case, interpretation of the outcome of the usual NHST is conditional on some model, and the hypothesis is stated in terms of parameters of the model.  It is due diligence for the analyst to ensure that the model is a defensible approximation to reality.  This is usually performed by examining graphical diagnostics of some quantities that arise from the model estimation. For example, the analyst might create histograms of the residuals, which are the differences between the observations and the values that would have been observed had the data followed the assumed model exactly.

The \emph{$p$-value} is defined as the probability under repeated equivalent experimentation that a result as or more extreme (relative to the null) would be observed, conditional on the null hypothesis being true. True $p$-values are never zero (with probability one). We can use $p$-values to perform null hypothesis testing by picking a threshold $\alpha$ and rejecting the null hypothesis when the $p$-value is below $\alpha$.
%
%

\subsection{Power of a Test}

The power of a null hypothesis test describes the number of type I and and type II errors in terms of the comparison of the p value with the threshold $\alpha$. A more powerful test is better at rejecting the null hypothesis when it indeed is false. The size of the test is the power evaluated at the null hypothesis.  For $H_0$ the null hypothesis, and $H_1$ the alternative we define the \emph{power} of the NHST by
$$\mbox{power} = \mathbb P\big( \mbox{reject } H_0 \big| H_1 \mbox{ is true} \big).$$

Even though we are adopting the Fisherian approach to NHST, the power concept is useful as it provides a means of comparing between test statistics.  In general, so long as the test is unbiased (i.e.\ the size is correctly achieved at the null hypothesis), we will prefer to use the test that is uniformly most powerful (UMP), or at least locally most powerful (LMP) in the region of the null hypothesis. In our case we use a permutation distribution to establish the behaviour of the test statistic under the null hypothesis, so formal considerations such as UMP and LMP are not possible.  Nonetheless, we can approximately assess different test statistics by using simulation experiments, and assessing the proportion of times that hypotheses are rejected at a nominal level based on specified differences.  We can compare these proportions directly between different test statistics to provide an idea of the relative merits of the test statistics.

\subsection{Randomization Tests}

Randomization-based tests relieve the analyst of the need to nominate a formal model under the null hypothesis, by providing an empirical estimate of the distribution of the test statistic under the null hypothesis.  That is, instead of nominating a theoretical distribution to use as a basis for comparison with the test statistic, an empirical null distribution is created, using simulation.  The procedure outlined in the following section is a randomization test. \citet{welsh-1996} provides a readable introduction.

\section{A Test Procedure}\label{sec:procedure}

\subsection{Two Sets of Labels (t-test)}

Assume that we have a collection of $n$ independent persistence
diagrams and a tentative labeling scheme that divides the collection
into two possibly dissimilar collections, say $\X_1$ containing $n_1$
diagrams and $\X_2$ containing $n_2$ diagrams.  For example, we may
conjecture that the persistence diagrams that represent fMRI data of
two groups of patients -- one group with a condition of interest, and
one without --- are dissimilar.  The assumption of independence
precludes the possibility that any of the observations may have an
influence on any of the other observations and is important for
generating the null distribution.  Our goal is to assess the strength
of evidence that the processes that generated the collections $\X_1$
and $\X_2$ differ.

We realize this goal in the NHST framework as follows.  We take as the
null hypothesis the claim that the labels are exchangeable; that is,
informally, that the current configuration of labels is no less likely
than would have happened under a random labeling scheme, relative to
the test statistic.  An example of this reasoning follows.  Given
three tosses of a fair coin, each possible configuration has an
identical probability --- $0.125$.  However, from the point of view of
counting the number of heads in three tosses, as a test statistic, it
is much less likely that the count will be three (for a fair coin,
$0.125$) than two (for a fair coin, $0.375$).  The same reasoning
holds in the proposed test: even though each possible configuration of
the label is equally possible under the null hypothesis, we conjecture
that very many of the random configurations lead to a value of the
test statistic that is quite different to that in the observed sample.

\subsection{Test statistics}

Randomization tests that are used to compare two numerical samples
usually focus on some function of the distance of the means of the
samples.  In the current study, computing the means is expensive,
therefore computing the distance from each observation to the means
for each simulated set of labels will also be expensive.  We therefore
instead nominate a function of the within-group pairwise distances as a test
statistic.  This statistic needs to be computed only once for each
possible pair, and be stored in a table.  Then simulation can proceed by
summing the distances of pairs of observations that are randomly
allocated to the same group.

When the observations are on the real line, and the measure of
location is obtained by minimising the L2 norm, the location estimate
is the mean, and the L2 norm is a monotonic function of the variance.
In the proposed setup we have two putative means and two putative
variances to consider.  The joint loss of any labelling scheme,
conditional on the sample sizes, can be expressed as the sum of the
group-wise variances.  Hence we propose that taking the mean or the
sum of the variances of the two groups would be a sensible test
statistic.  The usual expression for the sample variance (for sets of
real numbers), which is in the form closest to the L2 norm evaluated
at its minimum, is

%
\begin{equation}
\sigma^2_{\X} = \frac{1}{n-1} \sum_{i=1}^n (x_i- \bar x)^2  
\label{eqn:v.slow}
\end{equation}
however, an equivalent variation can be computed without first
calculating the mean, namely

\begin{equation}
\sigma^2_{\X} = \frac{1}{2n(n-1)} \sum_{i=1}^n \sum_{j=1}^n (x_i - x_j)^2  
\label{eqn:v.fast}
\end{equation}

One advantage of using (\ref{eqn:v.fast}) instead of
(\ref{eqn:v.slow}) is that the matrix of pairwise distances only has
to be computed once, and the means of the randomly generated samples
are not calculated.  Randomly shuffling the group labels amounts to
reading different sets of cells from the precalculated distance matrix.

For persistence diagrams with labeling $L$ into the sets $\X_1=\{X_{1,1}, X_{1,2}, \ldots, X_{1,n_1}\}$ and $\X_2=\{X_{2,1}, X_{2,2}, \ldots X_{2,n_2}\}$ the analogous test statistic is
\begin{equation}
\sigma^2_{\X_{12}}(L)= \sum_{m = 1}^2 
\frac{1}{2n_m(n_m-1)} \sum_{i=1}^{n_m} \sum_{j=1}^{n_m} d_2(X_{m,i} , X_{m,j})^2  
\label{eqn:test}
\end{equation}
where $d_2(\cdot, \cdot)$ is the distance function in \eqref{eq:distance}.

The distance between means is not a suitable test statistic for sets of persistence diagrams. There is a high computational cost of computing the means for each permutation. Furthermore, the Fr\'{e}chet mean is not necessarily unique which then leads to issues of how to define this loss function when it is not. 

%
%
%
%

\subsection{Families of joint loss functions}

The test statistic $\sigma^2_{\X_{12}}$ is not the only option for NHST permutation tests. It is an example of a joint loss function. A loss function is a function that maps data onto a real number intuitively representing some ``cost'' associated with the model and the data. A joint loss function is the sum of two or more loss functions corresponding to the total cost.

Just as there are many metrics there are many different joint loss functions. We can use other metrics to construct joint loss functions. Sensible choices include 
$$F_{p,q}(\{X_{1,i}\}, \{X_{2,i}\}) := \sum_{m=1}^2 \frac{1}{2n_m(n_m-1)} \sum_{i=1}^{n_m} \sum_{j=1}^{n_m} d_p(X_{m,i} , X_{m,j})^q$$
where $p\in [1,\infty]$, $d_p(\cdot, \cdot)$ is the distance function in \eqref{eq:distance} and $q \in [1,\infty)$.  If the grouping is sensible then these joint loss functions should be small. In this paper we will restrict our attention to $p=1,2, \infty$ and $q=1,2$.

%
%
%
%


\section{Monte-Carlo simulation of the $p$-value}\label{sec:Monte}
We can define a true $p$-value by considering how extreme the test statistic is on the observed data compared to the test statistics for every possible permutation of the labels. 

\begin{lemma}\label{lem:truep}
Let $X_{1,1}, X_{1,2}, \ldots, X_{1,n_1}$ and $X_{2,1}, X_{2,2}, \ldots X_{2,n_2}$ be persistence diagrams drawn i.i.d. (the null hypothesis) and let $\alpha$ be the proportion of all labelings $L$ such that $T(L)\leq T(L_{\text{observed}})$. Then for all $p\in [0,1]$ we have
$\mathbb{P}(\alpha \leq p)\leq p$.
\end{lemma}

\begin{proof}
Consider the list of all possible different sets of labels of the set $$\{X_{1,1}, X_{1,2}, \ldots, X_{1,n_1},X_{2,1}, X_{2,2}, \ldots X_{2,n_2}\},$$ alongside the observed labels. Order these labelings by the cost, lowest first, randomly arranging amongst ties. Let the random variable $W$ be the number of different labels appearing before the observed labels. Since under the model, the persistence diagrams are i.i.d., there is a uniform probability of the location of the original labeling over all the rankings. That is, $W$ has a uniform probability over the natural numbers from $0$ to $N$. This implies that $\P(W\leq k)= k$ for all $k \in \{0,1,\ldots N\}$. Furthermore $\P(W/N \leq p)\leq p$ for all $p\in [0,1]$.

Note that there is a coupling between $W/N$ and $Z$ with $W/N \leq Z$. They agree except potentially in the case where there are multiple labelings with the same cost as the originally observed. Using this coupling we conclude that for all $p \in [0,1]$
$$\mathbb{P}(Z \leq p) \leq \mathbb{P}(W/N\leq p) \leq p.$$
\end{proof}

The permutation $p$-value for a given labeling $L_0$ under the cost function $F$ is defined to be the proportion of labeling $L$ such that $F(L) \leq F(L_0)$. This is a true $p$-value by Lemma \ref{lem:truep}.

The total number of permutations of the labels is $\binom{n_1 + n_2}{n_1}$, which is usually far too large to check exhaustively.  We instead sample from the set of permutations with uniform probability. This will provide an unbiased estimator of the true permutation $p$-value. The randomization NHST algorithm is presented as Algorithm \ref{alg:unbiased}. 

\begin{algorithm}
 \SetAlgoLined
 \KwData{$n_1+ n_2$ persistence diagrams with labels $L_{\text{observed}}$ in disjoint sets of size $n_1$ and $n_2$, number of repetitions $N$, a joint loss function $F$}
 \KwResult{estimate of the total permutation p value}
 initialization - Z=0\;
 Compute $F(L_{\text{observed}})$ for the observed labels\;
 \For{$i=1$ to $N-1$}
 {Randomly shuffle the group labels into disjoint sets of size $n_1$ and $n_2$ to give labeling $L$\;
Compute $F(L)$ for the new samples\;
\If {$F(L)\leq F(L_{\text{observed}})$  }
{ $Z\mathrel{+}=1$}	
	}

$Z\mathrel{/}=N$\;
Output $Z$
\caption{An unbiased estimator of the permutation $p$-value.}\label{alg:unbiased}
\end{algorithm}

The output of Algorithm \ref{alg:unbiased} is an unbiased estimator of the permutation $p$-value but is not itself a $p$-value. This may at first appear counterintuitive. However, observe that the output of our estimator could be zero but a permutation $p$-value should never be zero. It must be at least $1/\binom{n_1 + n_2}{n_1}$ as $F(L_0) \leq F(L_0)$.
For more details as to why $Z$ is not a $p$-value see \cite{phipson2010permutation}. 

However, we can compute use the output of $Z$ to compute a true $p$-value. Let $p_{\text{total}}$ be the permutation $p$-value computed using all the permutations. The output $Z$ of Algorithm~1 follows the distribution 
$$Z \sim \frac{B(n,p_{\text{true}})}{n}.$$ 

\begin{thm}[\cite{phipson2010permutation}]
Let $Z$ be the output of Algorithm~1. $\dfrac{Z +1}{N+1}$ is a true $p$-value.
\end{thm}
The proof uses the facts that the usual distribution of the $p$-value under the null hypothesis is uniform and that $Z\sim \frac{B(n,p_{\text{true}})}{n}.$ For each randomly chosen permutation of labels $L$ there is a $p_\text{total}$ chance that $T(L)\leq T(L_{\text{observed}})$ and each of the label permutations is drawn independently.

By modifying  Algorithm \ref{alg:unbiased} we have obtain an algorithm for a true $p$-value.

\begin{algorithm}
 \SetAlgoLined
 \KwData{$n_1+ n_2$ persistence diagrams with labels $L_{\text{observed}}$ in disjoint sets of size $n_1$ and $n_2$, number of repetitions $N$, a joint loss function $F$}
 \KwResult{p value}
 initialization - Z=1\;
 Compute $F(L_{\text{observed}})$ for the observed labels\;
 \For{$i=1$ to $N-1$}
 {Randomly shuffle the group labels into disjoint sets of size $n_1$ and $n_2$ to give labeling $L$\;
Compute $F(L)$ for the new samples\;
\If {$F(L)\leq F(L_{\text{observed}})$  }
{ $Z\mathrel{+}=1$}	
	}

$Z\mathrel{/}=(N+1)$\;
Output $Z$
\caption{$p$-value algorithm for persistence diagrams.}\label{alg:true}
\end{algorithm} 
\section{Examples}\label{sec:examples}

\subsection{Point clouds of nearby shapes}

Now let $K$ be the circle of radius $1$ as shown in Figure \ref{fig:K}. Let $M$ be two concentric circles with radius $0.9$ and $1.1$ as shown in Figure \ref{fig:M}. We have the Hausdorff distance between $K$ and $M$ is $0.1$. 

\begin{center}
\begin{figure}[htp]
\begin{minipage}{0.45\linewidth}\centering
\begin{tikzpicture}
\draw[] (0,0) circle [radius=1];
\end{tikzpicture}
\caption{$K$}\label{fig:K}
\end{minipage}
\begin{minipage}{0.45\linewidth}\centering
\begin{tikzpicture}
\draw[] (0,0) circle [radius=0.9];
\draw[](0,0) circle [radius=1.1];
\end{tikzpicture}
\caption{$M$}\label{fig:M}
\end{minipage}
\end{figure}
\end{center}

Let $\rho_K(\sigma)$ denote the distribution of the convolution of the uniform distribution on $K$ with Gaussian noise $\mathcal{N}(0,\sigma^2)$. Similarly let $\rho_M(\sigma)$ denote the distribution of the convolution of the uniform distribution on $M$ with Gaussian noise $\mathcal{N}(0,\sigma^2)$. We want to compare sets of persistence diagrams constructed from point clouds drawn from $\rho_K(\sigma)$ to those constructed from point clouds drawn from $\rho_M(\sigma)$ for a range of $\sigma$.

We will now describe the simulation procedure. First we fixed a $\sigma \in \{0, 0.01, 0.02, 0.03, 0.04, 0.05\}$. Then we constructed 100 sets of 10 point clouds with each of the 10 point clouds consisting of $50$ i.i.d. points drawn from $\mu_K(\sigma)$. Similarly, we created 100 sets of 10 point clouds drawn using $\mu_M(\sigma)$. We then constructed the persistence diagrams in dimension $1$ for each of these 2000 point clouds via a Rips filtration (as described in \ref{subsec:DistPers}) and labelled each with either $K$ or $M$ depending on whether the point cloud was drawn using $\mu_K(\sigma)$ or $\mu_M(\sigma)$. These persistence diagrams were effectively summaries of the changes in the space of non-contractible loops as we considered the union of balls around each of the points in the point cloud as we increased the radius.

We then took $10$ persistence diagrams labelled with $K$ and $10$ persistence diagrams labelled with $M$. We then computed the corresponding $Z$ from Algorithm \ref{alg:true}, repeating the algorithm with the different joint loss functions $F_{(1,1)}, F_{(1,2)},F_{(2,1)},F_{(2,2)},F_{(\infty,1)}$ and $F_{(\infty,2)}$. Here, for each random labelling $L$, the joint loss functions are
$$F_{(p,q)}(L=(\{X_{1,i}\}, \{X_{2,i}\}))= \frac{1}{2\cdot 10 \cdot 9} \left(\sum_{i,j=1}^{10} d_p(X_{1,i} , X_{1,j})^q + \sum_{i,j=1}^{10} d_p(X_{2,i} , X_{2,j})^q.\right)$$ 

Using the set of $100$ values of $Z$ for the different simulations we estimated the power of the permutation test under these different joint loss functions and cutoffs of $\alpha=0.005, 0.01, 0.05$ and $0.1$. For a fixed $\sigma$, joint loss function and $\alpha$ we counted the number of corresponding $Z$ such that $Z<\alpha$ 

We ran this simulation for $\sigma=0, 0.01, 0.02, 0.03, 0.04$ and $0.05$. The results are tabulated in Figure \ref{fig:power}. Each figure corresponds to a different loss function and the different colours correspond to different levels of $\alpha$. As would be expected the ability to distinguish the two sets of diagrams increases as the sets are further apart and as the number of points in the point cloud increases.

An important observation is that the power varies amongst the different loss functions. In particular the loss functions using the bottleneck distance ($F_{(\infty, 1)}$ and $F_{(\infty, 2)}$) do much worse.

\begin{center}
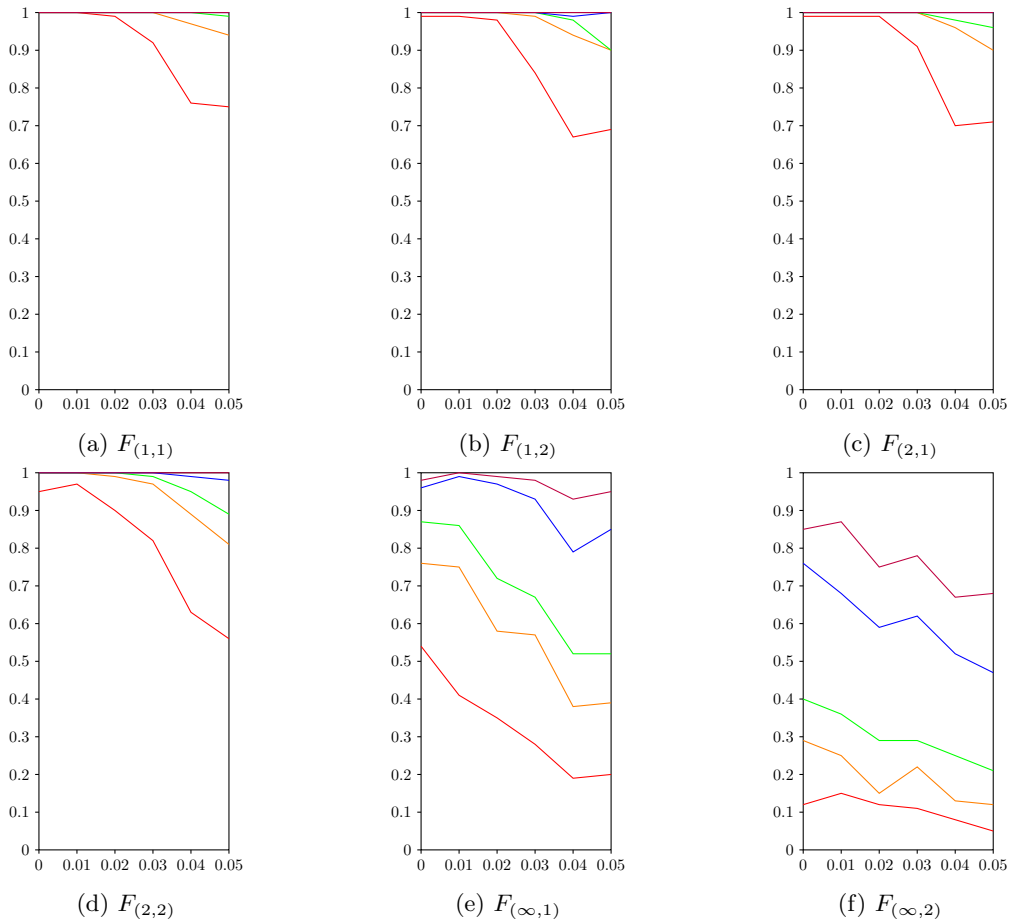
\begin{figure}[htp]
\begin{minipage}{0.33\linewidth}\centering
\begin{tikzpicture}[scale=.05,  every node/.style={scale=0.6}]
\draw[](0,0)--(0,-1)node[below]{$0$};
\draw[](0,90)--(-1,90)node[left]{$0.9$};
\draw[](0,80)--(-1,80)node[left]{$0.8$};
\draw[](0,70)--(-1,70)node[left]{$0.7$};
\draw[](0,60)--(-1,60)node[left]{$0.6$};
\draw[](0,50)--(-1,50)node[left]{$0.5$};
\draw[](0,100)--(-1,100)node[left]{$1$};
\draw[](0,10)--(-1,10)node[left]{$0.1$};
\draw[](0,20)--(-1,20)node[left]{$0.2$};
\draw[](0,30)--(-1,30)node[left]{$0.3$};
\draw[](0,40)--(-1,40)node[left]{$0.4$};
\draw[](0,0)--(-1,0)node[left]{$0$};
\draw[](10,0)--(10,-1)node[below]{$0.01$};
\draw[](20,0)--(20,-1)node[below]{$0.02$};
\draw[](30,0)--(30,-1)node[below]{$0.03$};
\draw[](40,0)--(40,-1)node[below]{$0.04$};
\draw[](50,0)--(50,-1)node[below]{$0.05$};
\draw[](50,0)--(0,0)--(0,100)--(50,100)--(50,0);
\draw[red](0,100)--(10,100)--(20,99)--(30,92)--(40,76)--(50,75);
\draw[orange](0,100)--(10,100)--(20,100)--(30,100)--(40,97)--(50,94);
\draw[green](0,100)--(10,100)--(20,100)--(30,100)--(40,100)--(50,99);
\draw[blue](0,100)--(10,100)--(20,100)--(30,100)--(40,100)--(50,100);
\draw[purple](0,100)--(10,100)--(20,100)--(30,100)--(40,100)--(50,100);
\end{tikzpicture}
\subcaption{ $F_{(1,1)}$}\label{fig:1,1}
\end{minipage}
\begin{minipage}{0.33\linewidth}\centering
\begin{tikzpicture}[scale=.05,  every node/.style={scale=0.6}]
\draw[](0,0)--(0,-1)node[below]{$0$};
\draw[](0,90)--(-1,90)node[left]{$0.9$};
\draw[](0,80)--(-1,80)node[left]{$0.8$};
\draw[](0,70)--(-1,70)node[left]{$0.7$};
\draw[](0,60)--(-1,60)node[left]{$0.6$};
\draw[](0,50)--(-1,50)node[left]{$0.5$};
\draw[](0,100)--(-1,100)node[left]{$1$};
\draw[](0,10)--(-1,10)node[left]{$0.1$};
\draw[](0,20)--(-1,20)node[left]{$0.2$};
\draw[](0,30)--(-1,30)node[left]{$0.3$};
\draw[](0,40)--(-1,40)node[left]{$0.4$};
\draw[](0,0)--(-1,0)node[left]{$0$};
\draw[](10,0)--(10,-1)node[below]{$0.01$};
\draw[](20,0)--(20,-1)node[below]{$0.02$};
\draw[](30,0)--(30,-1)node[below]{$0.03$};
\draw[](40,0)--(40,-1)node[below]{$0.04$};
\draw[](50,0)--(50,-1)node[below]{$0.05$};
\draw[](50,0)--(0,0)--(0,100)--(50,100)--(50,0);
\draw[red]((0,99)--(10,99)--(20,98)--(30,84)--(40,67)--(50,69);
\draw[orange](0,100)--(10,100)--(20,100)--(30,99)--(40,94)--(50,90);
\draw[green](0,100)--(10,100)--(20,100)--(30,100)--(40,98)--(50,90);
\draw[blue]((0,100)--(10,100)--(20,100)--(30,100)--(40,99)--(50,100);
\draw[purple](0,100)--(10,100)--(20,100)--(30,100)--(40,100)--(50,100);
\end{tikzpicture}
\subcaption{ $F_{(1,2)}$}\label{fig:1,1}
\end{minipage}
\begin{minipage}{0.33\linewidth}\centering
\begin{tikzpicture}[scale=.05,  every node/.style={scale=0.6}]
\draw[](0,0)--(0,-1)node[below]{$0$};
\draw[](0,90)--(-1,90)node[left]{$0.9$};
\draw[](0,80)--(-1,80)node[left]{$0.8$};
\draw[](0,70)--(-1,70)node[left]{$0.7$};
\draw[](0,60)--(-1,60)node[left]{$0.6$};
\draw[](0,50)--(-1,50)node[left]{$0.5$};
\draw[](0,100)--(-1,100)node[left]{$1$};
\draw[](0,10)--(-1,10)node[left]{$0.1$};
\draw[](0,20)--(-1,20)node[left]{$0.2$};
\draw[](0,30)--(-1,30)node[left]{$0.3$};
\draw[](0,40)--(-1,40)node[left]{$0.4$};
\draw[](0,0)--(-1,0)node[left]{$0$};
\draw[](10,0)--(10,-1)node[below]{$0.01$};
\draw[](20,0)--(20,-1)node[below]{$0.02$};
\draw[](30,0)--(30,-1)node[below]{$0.03$};
\draw[](40,0)--(40,-1)node[below]{$0.04$};
\draw[](50,0)--(50,-1)node[below]{$0.05$};
\draw[](50,0)--(0,0)--(0,100)--(50,100)--(50,0);
\draw[red](0,99)--(10,99)--(20,99)--(30,91)--(40,70)--(50,71);
\draw[orange](0,100)--(10,100)--(20,100)--(30,100)--(40,96)--(50,90);
\draw[green](0,100)--(10,100)--(20,100)--(30,100)--(40,98)--(50,96);
\draw[blue](0,100)--(10,100)--(20,100)--(30,100)--(40,100)--(50,100);
\draw[purple](0,100)--(10,100)--(20,100)--(30,100)--(40,100)--(50,100);
\end{tikzpicture}
\subcaption{ $F_{(2,1)}$}\label{fig:1,1}
\end{minipage}
\begin{minipage}{0.33\linewidth}\centering
\begin{tikzpicture}[scale=.05,  every node/.style={scale=0.6}]
\draw[](0,0)--(0,-1)node[below]{$0$};
\draw[](0,90)--(-1,90)node[left]{$0.9$};
\draw[](0,80)--(-1,80)node[left]{$0.8$};
\draw[](0,70)--(-1,70)node[left]{$0.7$};
\draw[](0,60)--(-1,60)node[left]{$0.6$};
\draw[](0,50)--(-1,50)node[left]{$0.5$};
\draw[](0,100)--(-1,100)node[left]{$1$};
\draw[](0,10)--(-1,10)node[left]{$0.1$};
\draw[](0,20)--(-1,20)node[left]{$0.2$};
\draw[](0,30)--(-1,30)node[left]{$0.3$};
\draw[](0,40)--(-1,40)node[left]{$0.4$};
\draw[](0,0)--(-1,0)node[left]{$0$};
\draw[](10,0)--(10,-1)node[below]{$0.01$};
\draw[](20,0)--(20,-1)node[below]{$0.02$};
\draw[](30,0)--(30,-1)node[below]{$0.03$};
\draw[](40,0)--(40,-1)node[below]{$0.04$};
\draw[](50,0)--(50,-1)node[below]{$0.05$};
\draw[](50,0)--(0,0)--(0,100)--(50,100)--(50,0);
\draw[red](0,95)--(10,97)--(20,90)--(30,82)--(40,63)--(50,56);
\draw[orange](0,100)--(10,100)--(20,99)--(30,97)--(40,89)--(50,81);
\draw[green](0,100)--(10,100)--(20,100)--(30,99)--(40,95)--(50,89);
\draw[blue](0,100)--(10,100)--(20,100)--(30,100)--(40,99)--(50,98);
\draw[purple](0,100)--(10,100)--(20,100)--(30,100)--(40,100)--(50,100);
\end{tikzpicture}
\subcaption{ $F_{(2,2)}$}\label{fig:1,1}
\end{minipage}
\begin{minipage}{0.33\linewidth}\centering
\begin{tikzpicture}[scale=.05,  every node/.style={scale=0.6}]
\draw[](0,0)--(0,-1)node[below]{$0$};
\draw[](0,90)--(-1,90)node[left]{$0.9$};
\draw[](0,80)--(-1,80)node[left]{$0.8$};
\draw[](0,70)--(-1,70)node[left]{$0.7$};
\draw[](0,60)--(-1,60)node[left]{$0.6$};
\draw[](0,50)--(-1,50)node[left]{$0.5$};
\draw[](0,100)--(-1,100)node[left]{$1$};
\draw[](0,10)--(-1,10)node[left]{$0.1$};
\draw[](0,20)--(-1,20)node[left]{$0.2$};
\draw[](0,30)--(-1,30)node[left]{$0.3$};
\draw[](0,40)--(-1,40)node[left]{$0.4$};
\draw[](0,0)--(-1,0)node[left]{$0$};
\draw[](10,0)--(10,-1)node[below]{$0.01$};
\draw[](20,0)--(20,-1)node[below]{$0.02$};
\draw[](30,0)--(30,-1)node[below]{$0.03$};
\draw[](40,0)--(40,-1)node[below]{$0.04$};
\draw[](50,0)--(50,-1)node[below]{$0.05$};
\draw[](50,0)--(0,0)--(0,100)--(50,100)--(50,0);
\draw[red](0,54)--(10,41)--(20,35)--(30,28)--(40,19)--(50,20);
\draw[orange](0,76)--(10,75)--(20,58)--(30,57)--(40,38)--(50,39);
\draw[green](0,87)--(10,86)--(20,72)--(30,67)--(40,52)--(50,52);
\draw[blue](0,96)--(10,99)--(20,97)--(30,93)--(40,79)--(50,85);
\draw[purple](0,98)--(10,100)--(20,99)--(30,98)--(40,93)--(50,95);
\end{tikzpicture}
\subcaption{ $F_{(\infty, 1)}$}\label{fig:1,1}
\end{minipage}
\begin{minipage}{0.33\linewidth}\centering
\begin{tikzpicture}[scale=.05,  every node/.style={scale=0.6}]
\draw[](0,0)--(0,-1)node[below]{$0$};
\draw[](0,90)--(-1,90)node[left]{$0.9$};
\draw[](0,80)--(-1,80)node[left]{$0.8$};
\draw[](0,70)--(-1,70)node[left]{$0.7$};
\draw[](0,60)--(-1,60)node[left]{$0.6$};
\draw[](0,50)--(-1,50)node[left]{$0.5$};
\draw[](0,100)--(-1,100)node[left]{$1$};
\draw[](0,10)--(-1,10)node[left]{$0.1$};
\draw[](0,20)--(-1,20)node[left]{$0.2$};
\draw[](0,30)--(-1,30)node[left]{$0.3$};
\draw[](0,40)--(-1,40)node[left]{$0.4$};
\draw[](0,0)--(-1,0)node[left]{$0$};
\draw[](10,0)--(10,-1)node[below]{$0.01$};
\draw[](20,0)--(20,-1)node[below]{$0.02$};
\draw[](30,0)--(30,-1)node[below]{$0.03$};
\draw[](40,0)--(40,-1)node[below]{$0.04$};
\draw[](50,0)--(50,-1)node[below]{$0.05$};
\draw[](50,0)--(0,0)--(0,100)--(50,100)--(50,0);
\draw[red](0,12)--(10,15)--(20,12)--(30,11)--(40,8)--(50,5);
\draw[orange](0,29)--(10,25)--(20,15)--(30,22)--(40,13)--(50,12);
\draw[green](0,40)--(10,36)--(20,29)--(30,29)--(40,25)--(50,21);
\draw[blue](0,76)--(10,68)--(20,59)--(30,62)--(40,52)--(50,47);
\draw[purple](0,85)--(10,87)--(20,75)--(30,78)--(40,67)--(50,68);
\end{tikzpicture}
\subcaption{ $F_{(\infty,2)}$}\label{fig:1,1}
\end{minipage}
\caption{The power of the permutation NHST with cut-offs $\alpha=0.001$(red), $0.005$ (orange), $0.01$ (green), $0.05$ (blue), and $0.1$ (purple). The $x$ coordinate indicates the noise $\sigma$ and the $y$-coordinate the simulated power.}\label{fig:power}
\end{figure}
\end{center}

\subsection{Distinguishing sets of shapes from silhouette databank}

In this example we will be using a variation on the the theme of a persistence diagram as a random element. Given a simplicial complex $M$ in Euclidean space and a unit vector $v$ we can create a filtration of $M$ by the height function $h_v$ in the direction of $v$ and hence we can construct a persistence diagram $X(K,v)$ from the filtration of $M$ by sublevel sets of $h_v$. The persistent homology transform of $M$ is the function from the sphere of directions to the space of persistence diagrams where $v$ is sent to $X(M,v)$. This process is explored in detail in \cite{turner2014persistent}. There it is shown that the persistent homology transform of a shape is a sufficient statistic and is stable under perturbations of the shape. As such it is reasonable, given sets of shapes, to analyze the sets of their persistent homology transforms. 

The distance squared between the persistent homology transforms of two shapes is effectively the integral over the unit sphere of the distances squared between the corresponding diagrams. This process can be made scale and translation invariant by appropriately modifying the diagrams $X(M,v)$. Furthermore it can be made rotation invariant by taking the infimum of all possible rotations. The $Lp$ distance between two aligned objects in the plane, $M_1$ and $M_2$, in the plane is 
\begin{align*}
d_{p}(M_1, M_2) := \left(\int_{S^{1}} d_p(X(M_1,v), X(M_2, v))^p\,dv\right)^{1/p}
\end{align*}
where $X(M_i,v)$ is the $0th$ dimensional persistence diagram corresponding to the height function $h_v:M_i \to \R, x \mapsto x\cdot v$.
The distance of unaligned objects is the minimal distance under different rotations. For more details the reader is referred to \cite{turner2014persistent}. \footnote{The reader should note that in \cite{turner2014persistent} the focus is on the L1 distance.} 

A shape database that has been commonly used in image retrieval is the 
MPEG-7 shape silhouette database \cite{Sikora01}. We used a subset of this database \cite{Lateckietal00} which includes seven class of objects:
 Bone, Heart, Glass, Fountain, Key, Fork, and Axe. There were twenty examples for each class for a total of 1400 shapes. The shapes are
 displayed in Figure \ref{images}. 
 
\begin{figure}[hbt]
\begin{center}
\includegraphics[height=1.5in]{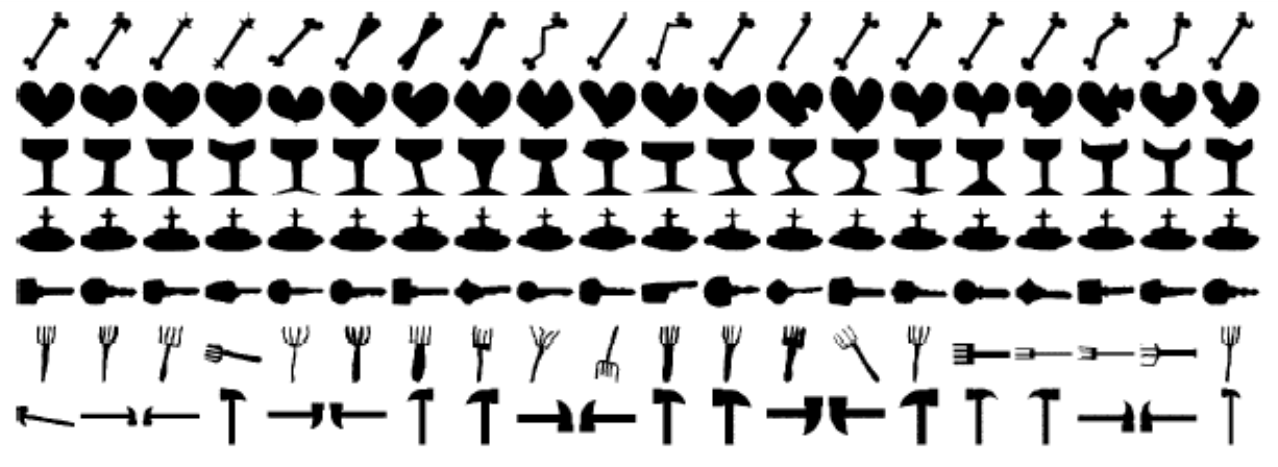}\\
\caption{\label{images}
The subset of the silhouette database. Each row corresponds to one of the objects: Bone, Heart, Glass, Fountain, Key, Fork, and Axe. Note that
although the objects are distinct, there is a great deal of variation within each object.
}
\end{center}
\end{figure}

We used the perimeters of the silhouettes which are available at \citet{Gao}. We applied the alignment algorithm we stated in Section 3.3 of  \cite{turner2014persistent} to shift, scale, and rotate the silhouettes. These perimeters are all homotopic to a circle so we used the $0$-th dimensional persistent homology transform with $64$ evenly spaced directions.

We considered both the loss functions 
$$F_1(\{X_{1,i}\}, \{X_{2,i}\}) := \sum_{m=1}^2 \frac{1}{2n_m(n_m-1)} \sum_{i=1}^{n_m} \sum_{j=1}^{n_m} d_1(X_{m,i} , X_{m,j})$$ and $$F_2((\{X_{1,i}\}, \{X_{2,i}\}) := \sum_{m=1}^2 \frac{1}{2n_m(n_m-1)} \sum_{i=1}^{n_m} \sum_{j=1}^{n_m} d_2(X_{m,i} , X_{m,j})^2.$$

For each pair of classes we then computed the corresponding $p$ values using $10000$ repetitions. The algorithm for each different pair always resulted in $p<0.0001$, using either the $F_{(1,1)}$ or the $F_{(2,2)}$ loss function. This implies that we expect that the distributions of persistent homology transforms are very significantly different. This result implies that NHST-based classification via the persistent homology transform should be possible.

\subsection{Concurrence Topology in fMRI data}

Given a set of variables and and samples of dichotomized data across those variables, concurrence topology is a method of creating filtration of a simplicial complex (and hence also persistence diagrams)  to reflect the frequency of when subsets of the variables are simultaneously active. This method is studied in \cite{ellisklein} and applied to fMRI data for both subjects diagnosed with ADHD and healthy controls. We briefly describe the procedure and refer the reader to \cite{ellisklein} for details.

Some set locations in the brain were measured. For each time interval we get a number associated to how active that location in the brain is. These data are dichotomized by choosing a cutoff value. For each location in the brain we associate a vertex $v_i$. We assign to the vertex $v_i$ the value of number of times that location was active. We assign to the edge $[v_i, v_j]$ the number of times both $v_i$ and $v_j$ were active simultaneously. Similarly assign to the face $[v_i, v_j, v_k] $ the number of times all three of $v_i$, $v_j$ and $v_k$ were active simultaneously. The same process assigns values to all simplices in the complete simplicial complex. The filtration is by superlevel sets. This is a simplification of the procedure. As part of the cleaning process some of the locations of in the brain are ignored and this set is different depending on the subject.

We calculated the p values associated with the sets of persistence diagrams in the ``default mode network'' that Ellis and Klein computed and kindly provided.
{\small
\begin{table}[ht]
\caption{Output of the Algorithm \ref{alg:true}, using loss function $F_{2,2}$, with $10000$ repetitions}
\centering
\begin{tabular}{c c c c c c c}
\hline\hline
Dimension & 0 & 1 & 2 & 3 & 4 & 5 \\ [0.5ex] 
\hline
ADHD vs Control 			& 	0.75272	&	0.20537	& 	0.50679 	&	0.41815	 &	0.10146	&	0.01162\\
ADHD vs Control in Females 	&  	0.68016	& 	0.59175	&	0.77673 	&	0.90267	&\red{0.00588}	&	0.30057	\\
ADHD vs Control in Males 	&	0.46101	&	0.22070 	& 	0.59409  	&	0.48437 	&	0.41364	&\red{0.00975}	\\
Females vs Males in Control 	&\red{0.00930}	& 	0.59964	&	0.33578	&	0.09851	&	0.19303	&	0.26304\\
Females vs Males in ADHD 	&	0.48694	&	0.45473	&	0.60937	&	0.59045 	&	0.02443	&	0.83618\\ [1ex]
\hline
\end{tabular}
\label{table:ADHD}
\end{table}
}
In \red{red} are the p values which are $\leq 0.01$. If we take a significance cutoff at $p=0.01$ then our expected false discovery rate is much less than $1$.  

A few comments should be made about the data set. The fMRI data set was generated at New York University and distributed as part of the 1000 Functional Connectomes project (\url{http://fcon 1000.projects.nitrc.org/}). It includes 41 healthy controls (NewYork a part1) and 25 adults diagnosed with ADHD ('NewYork a ADHD").
Unfortunately the samples were highly imbalanced with respect to age and gender. Only 20\% of the ADHD group was female, while about half of the controls were. About 25\% of the controls were children (younger than 20; median age = 12), while there were no children in the ADHD group. Among adults, ages ranged from about 21 to about 50 in each group. The median age in the ADHD group was 37, while in the control group the median adult age was 27. We did not control for age while computing the p values, and it is not presently clear how that could be done in this context.

\section{Discussion and Future Directions}

This paper shows an example of how non-parametric methods from statistics can be adapted for the use with topological summary statistics. There is much potential in exploring other non-parametric methods. 

There are a variety of related problems for hypothesis testing.  One problem is whether it is possible to do alternate hypothesis testing when the observations are persistence diagrams. Another area to explore is determining under what circumstances we can guarantee the power of the null hypothesis testing procedure to distinguish different distributions of persistence diagrams. This exploration may be through theoretical results or via simulations.

\subsection{More than Two Label Sets}

When the goal is to test whether $k > 2$ groups of observations
differ, we can use an extension of the test statistic in
equation~\ref{eqn:test} that is analogous to the F statistic in
analysis of variance.

\begin{equation}
\sigma^2_{x_{k}} = \sum_{m = 1}^k 
\frac{1}{2n_m(n_m-1)} \sum_{i=1}^{n_m} \sum_{j=1}^{n_m} (x_{mi} - x_{mj})^2  
\end{equation}

This is not the same as the F statistic because we do not propose to
compute the between-groups sums of squares, as that is expensive in
this setting.

\section{Acknowledgements}

We thank Steve Ellis and Arno Klein for providing us with the persistence diagrams produced in their work. The authors would like to acknowledge the assistance of the Defence Science Institute in facilitating this work. 

\bibliographystyle{apalike}
\bibliography{nhst}

\end{document}